\documentclass[copyright,creativecommons]{eptcs}
\providecommand{\event}{GandALF 2012} 
\usepackage[latin1]{inputenc}
\usepackage{multirow} 
\usepackage{amsmath}
\usepackage{amsfonts}
\usepackage{amssymb}
\usepackage{gastex}
\usepackage{stmaryrd}
\usepackage{framed}
\usepackage{wrapfig}
\def\statespace{M}

\usepackage{pstricks, pst-node, pst-tree}
\usepackage{graphicx}

\newtheorem{theorem}{Theorem}[section]

\newtheorem{lemma}[theorem]{Lemma}

\newtheorem{corollary}[theorem]{Corollary}

\newtheorem{construction}[theorem]{Construction}

\newtheorem{definition}[theorem]{Definition}

\newenvironment{proof}{ {\bf Proof: }}{}

\newcommand\cla{\mathsf{cone}}

\newcommand\suc{\mathsf{suc}}

\newcommand\frag{\mathsf{frag}}
\newcommand\res{\mathsf{res}}
\newcommand\safe{\mathsf{sfrch}}

\newcommand\qed{\hfill\ensuremath{\Box}}

\title{Rapid Recovery for Systems with Scarce Faults%
\thanks{The research was supported
by the National Science Council (NSF) 97-2221-E-002-129-MY3, 
by the Israeli Science Foundation (ISF) grant 1252/09, and 
by the Engineering and Physical Sciences Research Council (EPSRC) grant EP/H046623/1.}}
\def\titlerunning{Rapid Recovery for Systems with Scarce Faults}

\author{Chung-Hao Huang
\institute{Department of Electrical Engineering\\
           National Taiwan University, Taiwan, ROC}
\and
Doron Peled
\institute{Department of Computer Science\\ 
	   Bar Ilan University, Ramat Gan 52900, Israel}
\and
Sven Schewe
\institute{Department of Computer Science\\
	   University of Liverpool, Liverpool, UK}
\and
Farn Wang
\institute{Dept. of EE, 
           National Taiwan University, Taiwan, ROC\\
           CITI, Academia Sinica, Taiwan, ROC
           }
}

\def\authorrunning{Chung-Hao Huang, Doron Peled, Sven Schewe, and Farn Wang}

\usepackage{eepic}
\usepackage{pslatex}
\usepackage{graphicx}
\usepackage{latexsym}
\usepackage{amsfonts}

\usepackage{xcolor}

\begin{document}
\maketitle

\begin {abstract}
Our goal is to achieve a high degree of fault tolerance through the control of a safety critical systems.
This reduces to solving a game between a malicious environment that injects failures and a controller who tries to establish a correct behavior.
We suggest a new control objective for such systems that offers a better balance between complexity and precision:
we seek systems that are \emph{$k$-resilient}.
In order to be $k$-resilient, a system needs to be able to rapidly recover from a small number, up to $k$, of local faults infinitely many times, provided that blocks of up to $k$ faults are separated by short recovery periods in which no fault occurs.
$k$-resilience is a simple but powerful abstraction from the precise distribution of local faults, but much more refined than the traditional objective to maximize the number of local faults.
We argue why we believe this to be the right level of abstraction for safety critical systems when local faults are few and far between.
We show that the computational complexity of constructing optimal control with respect to resilience is low and demonstrate the feasibility through an implementation and experimental results.
\end{abstract}

\section{Introduction}

We study the problem of achieving fault tolerance by construction (i.e., by synthesis~\cite{Church/63/Logic,Pnueli+Rosner/89/Synthesis,Rabin/69/Automata}).
This is a challenging problem for two simple reasons.
One reason is that the size of the relevant systems is too large for exhaustive analysis.
A second difficulty is the need to preserve the behavior of the system under failures and the proposed recovery mechanism, related to (i.e., included in) the non-failing behavior of the {\em original system}.
This makes the problem  much harder than checking whether or not a system satisfies a given set of simple temporal properties.

Our goal is to develop a technique for synthesizing a controller that is resilient to an \emph{unbounded} number of failures.
With this problem in mind, we took our inspiration from methods for resilient avionic systems \cite{Ru}, where fault tolerance is designed to recover from a \emph{bounded} number of failures.
The number of failures a system needs to tolerate is inferred from the given maximal duration of a flight and the mean time between failures of the individual components.
Designing an optimal control with respect to this objective reduces to finding the maximal number of failures a system can tolerate without exhibiting an error.
Maximizing the number of tolerable failure is an abstraction from the control objective of minimizing the chance of erroneous system behavior.
This abstraction is justified by the simplicity of the control objective, but we argue that a finer abstraction provides for better control while maintaining a simple objective.

We develop a technique for synthesizing a controller that is resilient to an \emph{unbounded} number of failures. 
Instead, we assume a bound on the number of {\em dense failures}, which may occur before the system is fully recovered.
After full recovery, the system is allowed again the same number of failures.
Now, if the mean time between failures is huge compared to the time the system needs to fully recover, then the expected time for system breakdown grows immensely.
We assume an abstract representation of the system and recovery building blocks such as redundancy and majority checking among components.

\paragraph{Comparison.}
To demonstrate the difference between the objective to tolerate $k$ errors and $k$ dense errors, we exemplify the quality guarantees one obtains for a system (e.g., an aeroplane) with an operating time of 20 hours and a mean time between exponentially distributed local failures (MTBF) of 10 hours, assuming a repair time of 36 seconds.
\begin{center}
\begin{tabular}{cccccccc}
                 & $1$      &    $2$   &  $3$  & $4$ & $5$ & $6$ & $\ldots$ \\
$k$ failures       & $59.4\%$ & $33.3\%$ & $14.3$ & $5.3\%$ & $1.7\%$ & $0.5\%$ & $\ldots$ \\
$k$ dense failures & $0.2\%$ & $2\cdot 10^{-4}\%$ & $2\cdot 10^{-7}\%$ & $2\cdot 10^{-10}\%$ & $2\cdot 10^{-13}\%$ & $2\cdot 10^{-16}\%$ & $\ldots$
\end{tabular}
\end{center}
The figures for $k$ failures are simply the values for the Poisson distribution with expected value $2$ ($20$~hours flight and $10$ hours MTBF).
To explain the figures for $k$ dense errors, consider the density of 2 dense errors occurring in close succession.
If a failure occurs, the chance that the next failure occurs within the repair time (36 seconds) is approximately $\frac{1}{1000}$.
Only failures for which this is the case could cause an error.
The mean time between blocks of two dense failures is therefore not ten hours, but ten thousand hours.
Likewise, it increases to ten million and ten billion hours for blocks of three and four dense failures, respectively.
Maximizing the number of dense failures that are permitted is therefore a natural design goal.

Our proposed correctness criterion is formalized through a game between the system and an environment that induces the failures; the system is 
trying to recover, while the environment can inject further failures.
We provide a technique, based on finding a game strategy,
that allows checking whether or not we can control our system 
to meet this predefined resilience level. It also allows finding the 
maximal resilience level of the system. Moreover, the sought strategy
provides the control that can be added to the system to achieve
this resilience level.
This is a more ambitious control problem than just meeting predefined bounds, 
allowing us to tap the full potential of a system.

\paragraph{\bf Related Work.}
Dijkstra's self-stabilization criterion~\cite{Arora,Dijkstra} suggests to build systems that eventually recover
to a `good state', from where the program commences normally.
Instead of {\em constructing} a system to satisfy such a goal, one
may want to apply control theory to {\em restrict} the execution of
an existing system to achieve an additional goal.
Our control objective is a recovery mechanism for up to $k$ failures.
After recovery, the system has to tolerate up to $k$ failures again, and so forth.
In this work, 
we suggest a mechanism to synthesize a recovery mechanism for a given fault model and recovery primitives.

Synthesis goes back to Church's solvability problem~\cite{Church/63/Logic}. It inspired Rabin's work on finite automata over infinite structures~\cite{Rabin/69/Automata} and B\"uchi and Landweber's works on finite games of infinite duration~\cite{Buchi/62/Automata,Buchi+Landweber/69/Sequential}.
These
techniques have been developed further in open
synthesis~\cite{Pnueli+Rosner/89/Synthesis,Pnueli+Rosner/89/Asynchronous,Kupferman+Vardi/97/Incomplete,Schewe+Finkbeiner/06/Asynchronous} and control theory \cite{AAE,EKA,GR,Ru}.

Traditionally, fault tolerance refers to various basic fault models~\cite{AAE}, such as a limited number of errors \cite{GRT}.
These traditional fault models are subsumed by more general synthesis or control objectives \cite{Asarin+Maler+Pnueli/94/controlerSynthesis,AAE,Ramadge+Wonham/89/Control,Thomas/94/finiteStrategies}; as simple objectives with practical relevance, they have triggered the development of specialized tools~\cite{EKA,GR}.

In \cite{DHLN}, an interesting notion of robustness based on Hamming 
and Lewenstein distance related to the number of past states is defined. 
It establishes a connection between these distances with a notion of 
synchronization that characterizes the ability of the system to reset for 
combinatorial systems.
In \cite{BGHJ}, `ratio games' are discussed, where the objective is 
to minimize the ratio between failures induced by the environment and system errors caused by them.

Maximizing the number of dense failures can, just like maximizing the overall number of failures, be viewed as an abstraction of the objective to minimize the likelihood of an error.
In principle, one could seek optimal control against a Markovian failure model rather than approaching an abstraction.
But approximating optimal control seems to expensive \cite{FRSZ-fsttcs11} and results in more complex control strategies.
Two further problem are that the precise probabilistic behavior of the environment is usually unknown, and that probabilistic analysis does not mix well with abstraction.

\section{Preliminaries}
\label{sec:faulTolerance}

We propose a method for achieving fault tolerance under an architecture that 
allows generic building blocks such as redundancy and majority 
checks \cite{Pradhan96}, while making some realistic 
restricting assumptions on the behavior of the faults.  
The original system can be described using a transition system or an automaton.
We impose a failure model that transforms the system by allowing new
{\em fault transitions}. We include a recovery mechanism that transforms the 
resulting faulty version of the system by allowing new
{\em recovery transitions}.

We assume that the transitions of the recovery mechanism
cannot block the added fault transitions. 
They refer to standard recovery building blocks, 
e.g., providing redundancy and performing majority checks.
We do not study or try to design these building blocks themselves, 
but resolve the choices that need to be made in the presence of alternative 
recoveries, e.g., which of them should be initiated.
The recovery mechanism itself
does neither take the number of the faults into account, nor their density.

Our task is to prune the choices available to the recovery
mechanism, based on assumptions made on the behavior of the fault model. 
One can consider this as a game between two players: the recovery
mechanism that needs to make the right decision when choosing from several
recovery transitions, against the failure model that can inject some
(restricted) amount of fault transitions into the system.

The ability to choose makes the
moves of the recovery mechanism {\em controllable}, while the moves
represented by the failure model are {\em uncontrollable}. 
Although the moves of the antagonistic failure mechanism are
uncontrollable, we may make some reasonable assumptions about them, e.g.,
limiting the number of failures occurring in close succession. 
Similar restrictions are
used in practice: in avionics, for example, one can provide some
measurable bounds for such parameters per aircraft, based on
capabilities such as flight distance and speed, and then 
add some safety factor.

Performing a game strategy check on the level of
the system combined with the fault model and recovery mechanism
is unrealistic. In fact, it is
also unnecessary:
standard building blocks, e.g., redundancy and majority-checking
for fault tolerance, provide a simple abstraction of it
into a finite state system.
It includes in particular some {\em error} states
that are unrecoverable, represented as sinks in the state graph of 
the system.  
The abstraction is often independent of the actual details of 
the original system.

A game theoretical search can then be applied to the abstract
version.
The goal of the protagonist is to avoid the error states; the fault tolerance
mechanism must ensure this goal in the presence of uncontrolled
(malicious) failures imposed by the failure mechanism.
The game takes the restrictions on
the faults into account. A win of the protagonist (the recovery mechanism)
provides a strategy, which can immediately be translated into a controller that can be used to restrict the behavior of the recovery transitions.
Playing the game on the actual system rather than the abstract one
according to this strategy
guarantees the same goal
for the actual system with the fault model and recovery mechanism.
We can not only apply such a strategy search in order to check
the resilience level: the strategy
obtained provides a control for the recovery mechanism that achieves this
resilience level.
Note that the controller can only block choices of the
recovery mechanism and not the transitions of the
original system (or the fault model).

The modeling of the transformation involves the fault model and the recovery
mechanism. The abstractions used to reason on a small model 
are rather standard model-checking techniques, hence 
would not be formally presented here. 
In the following section, we reason formally on the model that represents 
the abstracted system including the fault model and the recovery 
mechanism. 

\paragraph{\bf Running Example: Controlling Avionic Systems.}

Consider an example of a system that
includes $n$ processors, each can follow the instructions
of the original system $\statespace$, or be engaged in memory recovery. 
There are also $m$ copies of the memory.
When a copy of the memory fails, a processor can be assigned to recover it.
Majority check can be used to detect that a processor is faulty
or that memory is defected (usually, both
would happen at the same time). 
For recovery, we can
set a free processor to recover some memory, or
make a processor follow the code of the majority of processors.
The problem of the controller is to decide whether to make
a processor follow the execution of the majority, or to
assign it to recover faulty memory. 
We assume the following uncontrollable fault and controllable recovery transitions:
\begin{description}
\item
 Uncontrollable failure events:
   \begin{itemize}
   \item A processor performing the code of $\statespace$ fails. 

\item Memory fault occurs. 

\end{itemize}
\item Controllable recovery events:
\begin{itemize}
    \item Remove a process from following
    the main system $\statespace$ and being subjected to
    majority check.
    \item Assign a free processor to recover a memory
    that is known to be faulty.
    \item 
    Assign a free processor to follow the code of
    the majority processors.
    \end{itemize}
\end{description}
If there are no processes to perform any of the above controllable recovery
events, then the system enters immediately an unrecoverable 
failure state.

The recovery mechanism in the above is quite typical in 
the design of fault-tolerant systems \cite{Pradhan96}.  
As explained, a practical recovery mechanism usually 
does not rely on the detailed structure of $M$.  
Instead, fault-detection 
techniques such as parity checks, voting (for majority checks), 
etc., are usually employed.
They are generic and therefore independent of the structure of $M$.  
Thus, an abstraction of $M$ can be used in verifying the resilience 
of the recovery mechanism to dense faults. 
Such an observation can greatly reduce the cost of verifying 
the recovery mechanism.

We assume that the events of the recovery mechanism and fault detection are reliable.
The issue of how to make these events reliable is orthogonal to the problem we are solving in this paper.

\section{Resilience to $k$-dense Failures}
\label{sec:kresil}

We propose as a correctness criterion that can refer to a small abstract
model $\cal T$, which contains the recovery and fault actions
as its controllable and uncontrollable transitions, respectively.
Such a \emph{transition system with failures} ${\cal T} = (S,\iota,\tau_c,\tau_u,F)$ has 
\begin{itemize}
\item states $S$, with an initial state $\iota\in S$ and  error (i.e., \emph{final} or \emph{sink}) states $F\subseteq S$,
\item {\em controlled} (recovery) transitions 
$\tau_c \subseteq (S \smallsetminus F) \times (S \smallsetminus F )$ 
and
\item
{\em uncontrolled} (failure) transitions $\tau_u \subseteq 
(S \smallsetminus F) \times S$, where
$\tau = \tau_c \cup \tau_u$.
\end{itemize}
We denote with $\suc_c(s)=\{t \in S \mid (s,t) \in \tau_c\}$ and $\suc_u(s)=\{t \in S \mid (s,t) \in \tau_u\}$ the controlled and uncontrolled successors of a state $s$, respectively. The size of $\cal T$, denoted $|\cal T|$, is $\max\{|S|,|\tau|\}$.
Every (non-error) state in $S \smallsetminus F$ has at least
one controlled successor.

A state $s$ is $k$-resilient if the system can be controlled such that groups of up to $k$ dense failures can be tolerated infinitely many times, provided that the system is given enough time to recover.
We describe informally how $k$-resilience is tested using a two-stage game over $\cal T$ between the protagonist, i.e., the recovery mechanism, 
and the antagonist, i.e., the failure mechanism, who jointly 
move a pebble over the transition system.
The construction of this set is described in Construction \ref{con:res}.

Let $G  \subseteq S \smallsetminus F$ be a subset of the non-error states.
These states are intuitively the ``good region'' that the protagonist player wishes to remain in, and, if left, wishes to return to.
In order to get a good intuition for the objective of this game, it is useful to consider its behavior when using the set of $k$-resilient states as $G$.
The game is, however, defined for any set $G  \subseteq S \smallsetminus F$; we assume for the moment that this subset is given, and will later show how to calculate it.

The pebble is initially placed on a state $s \in G$.
In every move, the protagonist starts with selecting a controlled transition originating from the state on which the pebble is.
The antagonist can either agree on executing this transition (moving the pebble along 
it) or play a failure by selecting an uncontrollable transition, which is then executed.

The first part of the game is a safety part, which ends when the antagonist selects his first uncontrollable transition.
With this move, the game proceeds to a reachability part. 
In the safety part of the game, the protagonist must remain in $G$; she loses (and the antagonists wins) immediately if $G$ is left.
In the second phase, the protagonist becomes a reachability player with the goal to recover to $G$.
She wins immediately when she has reached $G$ again.
The antagonist is further restricted in that he can, overall, play no more than $k$ failures.
Of course, if a state in $F$ is reached, the antagonist wins immediately, as these states are sinks and are never in $G$.
In \emph{infinite} plays, the protagonist wins if the play stays for ever in the safety phase; otherwise the antagonist wins.
The set of states from which the protagonist wins this game is the set of $k$-resilient states, denoted $\res_k(G)$.

A state $s$ is called \mbox{$k$-sfrch} (where sfrch refers to the combined safety/reachability objective) with respect to 
a set $G\subseteq S \smallsetminus F$ of non-error states, denoted $s\in\safe_k(G)$, if there 
is a strategy for the protagonist to win this game.
This can be defined as follows:

\begin{definition}
\label{k-safe}
Let $\safe_k : 2^{S \smallsetminus F} \rightarrow 2^{S \smallsetminus F}$ be the function 
such that $\safe_k ( G )$ is the subset of $G$, from which the protagonist wins in the above game.
\end{definition}

\begin{wrapfigure}{r}{0.45\textwidth}
\begin{center}
\psset{xunit=20mm,yunit=7mm}
\begin{pspicture}(0,-.15)(3,1)
\rput(0,0){\circlenode{1}{$1$}}
\rput(1,0){\circlenode{2}{$2$}}
\rput(2,0){\circlenode{3}{$3$}}
\rput(3,0){\circlenode[doubleline=true]{4}{$4$}}

\nccircle{->}{1}{.35}
\nccircle{->}{2}{.35}
\nccircle{->}{3}{.35}
\ncarc[linecolor=red,linestyle=dashed]{->}{1}{2}
\ncarc{->}{2}{1}
\ncarc[linecolor=red,linestyle=dashed]{->}{2}{3}
\ncarc{->}{3}{2}
\ncline[linecolor=red,linestyle=dashed]{->}{3}{4}
\end{pspicture}
\end{center}
\caption{\label{fig:example} An example for calculating $\safe_k$}
\end{wrapfigure}
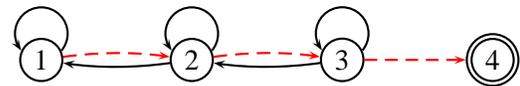

As an example for $k$-sfrch-ty, consider the transition system with four states, 
including a single error state
(state $4$, marked by a double line) shown in Figure~\ref{fig:example}.
The controlled transitions are depicted as black full arrows, the uncontrollable (or fault) transitions are depicted as {\color{red} dashed arrows}.
For $G=S\smallsetminus F=\{1,2,3\}$, all states in $G$ are in $\safe_0(G)$.
For all $k \geq 1$, we have $\safe_k(G)=\{1,2\}$:
the protagonist can simply stay in $\{1,2\}$ during the safety phase of the game, and once the antagonist plays a failure transition, the game progresses into the reachability phase, where the reachability objective is satisfied immediately.
This outlines the difference between $k$-sfrch-ty and the linear time property
of being able to tolerate $k$ failures, which would be satisfied by state $1$ and $2$ only for $k\leq 2$ and $k \leq 1$, respectively.

This difference raises the question if the rules of our game are depriving the antagonist of some of the $k$ failures he should intuitively have.
The answer is that this is not the case if we use the $k$-resilient states (or, more generally, any fixed point of $\safe_k$) as $G$.
In this case, we would win again from the state we reached; instead of depriving the antagonist, one could say that we reset the number of failures he can play to $k$.

For a state to be in $\safe_k(G)$, the system has a strategy to recover to $G$, given that a long enough execution commenced without another failure happening.
We say that two successive failures are in the same \emph{group of dense failures} if the sequence of states separating them was not long enough for recovery in the respective safety/reachability game.
Vice versa, if two successive failures are far enough apart such that the protagonist can guarantee recovery in this game, then they do not belong to the same group.
In order to define $k$-resilience, we need to find a set of states $G$ such
that recovering to $G$ by the protagonist (the fault tolerance mechanism)
is always possible, provided that at most $k$ failures occurred.
To obtain this, observe that nesting 
$\safe_k$ once, i.e., $\safe_k(\safe_k(\cdot))$, corresponds
to tolerating up to two sets of up to $k$ close errors, and so forth.
Thus, $k$-resilience is simply the greatest fixed point of the operator
$\safe_k$ from Definition~\ref{k-safe}:
a state is $k$-resilient if it is in $\safe_k(\safe_k(\safe_k(\ldots\ \safe_k(G\big]$, using sufficiently deep nesting that a fixed point is reached.
For the control strategy, it suffices to use the control strategy from the outermost $\safe_k$.

\begin{lemma} 
\label{fixpoint}
$\safe_k$ has a greatest fixed point.
\end{lemma}

\begin{proof}
Follows from the facts that
the function $\safe_k$ is monotonic
($G \subseteq G'$ implies $\safe_k ( G ) \subseteq \safe_k ( G' )$ because a winning strategy for the protagonist for $G$ is also a winning strategy for $G'$ for all states in $\safe_k(G)$) and operates
on a finite domain.
\qed
\end{proof}

As an example, consider again the transition systems with four states.
For $G=\{1\}$ ($\{1\}=\res_2(\{1,2,3\})$), the only state in $G$, state $1$, is $2$-resilient: it can recover with the recovery strategy to always go to the left.

\begin{construction}
\label{con:res}
$\res_k(G_0)$, for $G_0 = S \smallsetminus F$,
can be constructed by choosing $G_{i+1} = \safe_k(G_i)$ and fixing
$\res_k(G_0) = G_\infty$ to be the greatest fixed point of this construction.
\end{construction}
Note that this fixed point is what we are really interested in, while $\safe_k$ is a technical construction.
If this greatest fixed point $G = \safe_k ( G )$ is non-empty, the protagonist's strategy for 
the fixed point $G$ (guaranteeing eventual recovery to a state in $G$ within no more than $k$ failures, i.e., $k$-resilience) can be used to control the recovery mechanism, constraining its transitions to follow its winning strategy.

The natural control problem is to find optimal control that starts in the initial state $\iota$ of the transition system with failures.

\begin{definition}
For a non-error state $s \in S \smallsetminus F$, the \emph{resilience level $k_{\max}$} is the maximal $k$ such that $s$ is $k$-resilient.
\end{definition}

To illustrate this point on our running example, suppose we have $2k+1$ copies 
of the system, with the ability to perform majority checks
and identify the bad processes. According to the first, na\"ive
solution, there is nothing to do after $k$ failures. The
majority checks is still capable to maintain the correctness
of the combined behavior to follow the design of the original system $\statespace$.
However, there is no expectation that the system will be able
to recover at any point in the future into a situation where it is again
$k$-resilient; it {\em will} fail at the next round of faults.
Our dense fault tolerance criterion requires that, given no more failures
for enough time to allow recovery, the system will 
eventually recover to $k$-resilience again.

\subsection{Construction of $\safe_k$}
\label{ssec:kSafety}

Let $\safe_0(G) = \bigcup \{G' \subseteq G \mid  
\forall g \in G'\ \exists s \in G'.\ (g,s)\in \tau_c \}$. This
is the usual safety kernel of $G$ consisting of states from which there is
an infinite controlled sequence. It can be computed by the usual 
greatest fixed point construction.

\begin{lemma}
\label{lem:NLSafe0}
$\safe_0(G)$ can be constructed, together with a suitable memoryless
control strategy, in time linear in $(S,\tau_c)$, and testing if a state is in $\safe_0(G)$ is NL-complete.
\end{lemma}

\begin{proof}
Linear time is obvious for standard constructions of (least and)
greatest fixed points, and for the memoryless control 
strategy it suffices to stay in $\safe_0(G)$.

NL completeness can be shown by reduction to and from the repeated ST-reachability \cite{Papadimitriou/94/complexity} (the question whether there is a path from a state S to a state T and from T to itself in a directed graph).
\qed
\end{proof}

An intermediate step for the construction of $k$-sfrch states is an 
attractor construction that stays, through controlled moves, in a 
subset $L\subseteq S\smallsetminus F$ of non-error states. As only
controlled moves are allowed, this is merely a backwards reachability cone.

\begin{definition}
The \emph{controlled limited attractor set} of a set $G$ for a
limited region $L\subseteq S$, denoted $\cla_L(G)$ is the set 
$\cla_L(G) = \bigcap \{A \supseteq G \mid 
\forall s \in L.\ (\suc_c(s) \cap A \neq \emptyset ) \mbox{ implies } s \in A\}$ for which there is a strategy to move to $G$ without leaving $L$.
\end{definition}

The controlled limited attractor set 
$\cla_L(G)$ can be constructed using simple backwards reachability 
for $G$ of controlled transition through states of $L$.
\begin{lemma}
\label{lem:NLcla}
$A = \cla_L(G)$ and (a memoryless) attractor strategy for the states 
in $A\smallsetminus G$ towards $G$ can be constructed in time linear in 
the size of $(S,\tau_c)$. Determining whether a state is in $A$ 
is NL-complete (see \cite{Papadimitriou/94/complexity}). \qed
\end{lemma}

The controlled limited attractor set is used in the construction 
of $\safe_k(G)$.  We
further construct a descending chain $A_0 \supseteq A_1 \supseteq \ldots
\supseteq A_{k-1}$ of limited attractors $A_i$. From $A_i$ we have an
attractor strategy towards $G$ for the protagonist, which can tolerate 
up to $i$ further failures.
The respective $A_i$ are attractors that avoid error states 
and, for $i>1$, any uncontrolled transition leads to $A_{i-1}$

\begin{definition}
A state $s\in S$ is \emph{fragile} for a set $B\subseteq S$ 
if at least one of its uncontrolled successors is in $B$.
The set of fragile states for $B$ is $\frag(B) =
\{s \in S \mid \exists b \in B.\ (s,b)\in \tau_u\}$.
\end{definition}

The set $\frag(B)$ is easy to construct.
The limited regions $L_i$ of states allowed when approaching $G$ also form a descending chain $L_0 \supseteq L_1 \supseteq \ldots \supseteq L_k$.
Using these building blocks, we can compute the $k$-sfrch states:

\begin{construction}
\label{con:safe}
Starting with $L_0 = S \smallsetminus F$ (if no further failures are allowed for, only error states are disallowed on the way to $G$), we define the $A_i$'s and $L_i$'s recursively by
\begin{itemize}
\item $A_i=\cla_{L_i}(G)$ and

\item $L_{i+1}= L_0 \smallsetminus \frag(S\smallsetminus A_i)$,
\end{itemize}
and choose $\safe_k(G) = \safe_0(G \cap L_k)$.
\end{construction}

\noindent {\bf Explanation.}
The states in $L_{i+1}$ are the non failure states from which all 
(i.e., fault) uncontrolled transitions lead to a state in $A_i$.
The sets $A_i$ contain the states from which there is a
controlled path to $G$ that progresses through $L_i$; all uncontrolled transitions originating from any state of this path lead to $A_{i-1}$.
$A_0$ is therefore just the set of states from which there is a controlled path to $G$.

From all states in $A_{k-1}$, the protagonist therefore has a winning strategy in the second phase of the game described earlier:
if the antagonist can play at most $k-1$ failures, then the protagonist can make sure that $G$ is reached.

Finding a control strategy for $k$-sfrch control within $\safe_k(G)$ is simple:
as long as we remain in $\safe_k(G)=\safe_0(G \cap L_k)$, we can choose any control action that does not leave $\safe_k(G)$.
Once $\safe_k(G)$ is left through an uncontrolled transition to $A_{k-1},
A_{k-2}, ...$, we determine the maximal $i$ for which it holds that we are in $A_i$ and follow the attractor strategy of $\cla_{L_i}(G)$ towards $G$.

\subsection{Complexity}
All individual steps in the construction 
(intersection, difference, predecessor, and attractor) are linear in the 
size of the transition system with failures, and there are $O(k)$ of these 
operations in the construction. This provides a bi-linear (linear in $k$ and $|\cal T|$) algorithm for the construction of $\safe_k$ and a strategy for the protagonist:

\begin{lemma}
\label{lem:costSafe}
A memoryless control strategy for the states in $\safe_k(G)$ can be
constructed in time linear in both $k$ and the 
size $|\cal T|$ of the transition system with failures $\cal T$.
\qed
\end{lemma}

The complexity of determining whether or not a state $s$ is in $\safe_k(G)$ depends on whether or not we consider $k$ to be a fixed parameter.
Considering $k$ to be bounded (or fixed) is natural in our context, because $k$ is bounded by 
the redundancy.

\begin{lemma}
\label{lem:NLSafek}
For a fixed parameter $k$, testing if a state $s$ of a transition system
${\cal T} =(S,\iota,\tau_c,\tau_u,F)$ is in $\safe_k(G)$ is NL-complete.
\qed
\end{lemma}

The proof is based on an inductive argument that uses the closure of NL under complementation \cite{Immerman/88/NCoNL} in the induction step.
The details are moved to an appendix.

If $k$ is considered an input, then reachability in AND-OR graphs can easily be encoded in LOGSPACE:
It suffices to use the nodes of an AND-OR graph as the states, the outgoing edges of AND and OR nodes as controllable and uncontrollable edges, respectively, and then add a self loop for each OR node to the set of controllable edges.
Choosing $k$ to be the number of nodes of the AND-OR graph and $F$ to be the target nodes of the AND-OR graph, a state in the AND-OR graph is not reachable from a designated root node, if the respective state in the resulting transition system is $k$-sfrch.

Given that reachability in AND-OR graphs is PTIME-complete \cite{immerman/81/reachability}, this provides:

\begin{lemma}
\label{lem:ptc}
If $k$ is considered an input parameter, then testing if a state $s$ of a transition system with failures $\mathcal T=(S,\iota,\tau_c,\tau_u,F)$ is in $\safe_k(G)$ is PTIME-complete.
\qed
\end{lemma}

The construction of $\res_k(G)$ uses the repeated execution of $\safe_k(\cdot)$.
The execution of $\safe_k(\cdot)$ needs to be repeated at most $O(|G|)$ times until a fixed point is reached, and each execution requires at most $O(k\cdot|\cal T|)$ steps by Lemma~\ref{lem:costSafe}.

For the control strategy, we can simply use the control strategy from $\safe_k(G_\infty)$ from the fixed point $G_\infty=\res_k(G)$.
This control strategy is memoryless (cf.\ Lemma \ref{lem:costSafe}).
\begin{lemma}
\label{lem:costRes}
$\res_k(G)$ and a memoryless $k$-resilient control strategy 
for $\res_k(G)$ can be constructed in $O(k\cdot |G| \cdot |\cal T|)$ time.
\qed
\end{lemma}

The complexity class is (almost) independent of the parameter $k$:

\begin{lemma}
The problem of checking whether or not a state $s$ is $k$-resilient for a set $G$ is PTIME-complete for all $k>0$ and NL-complete for $k=0$.
\end{lemma}

\begin{proof}
For inclusion in PTIME, see the previous lemma. (Note that we can assume $k \leq |S|$ without loss of generality.)
For hardness in the $k>0$ case, we can use the same reduction from the reachability problem in AND-OR graphs as for $k$-sfrch-ty.

For $k=0$, $\safe_0(G)=\safe_0\big(\safe_0(G)\big)$ implies $\res_0(G)=\safe_0(G)$.
The problem of checking if a state is in $\res_0(G)$ is therefore NL-complete by Lemma \ref{lem:NLSafe0}.
\qed
\end{proof}

Finding the resilience level $k_{\max}$ for the initial state $\iota$
requires at most $O(\log k_{\max})$ many constructions of $\res_i(G)$.
We start with $i=1$, double the parameter until $k_{\max}$ is exceeded, and then use logarithmic search to find $k_{\max}$.

\begin{corollary}
For the initial state $\iota$, we can determine the resilience level $k_{\max} =
\max\{n\in \mathbb N_0 \mid \iota \in \res_n(S \smallsetminus F)\}$ of $\iota$,
$\res_{k_{\max}}(S \smallsetminus F)$, and a memoryless $k_{\max}$-resilient control
strategy for $\res_{k_{\max}}(S\smallsetminus F)$ in $O(|S| \cdot |{\cal T}|\cdot k_{\max} \log k_{\max})$ time.
\qed
\end{corollary}

\subsection{Handling Recovery Delay
\label{subsec.a.ra}
}

To keep the representation concise, we distinguish three types of transitions instead of two:
failure, control, and repair transitions.
Repair transitions are denoted $\tau_r \subseteq (S \smallsetminus F ) \times (S \smallsetminus F)$ and $\tau= \tau_c \cup \tau_r \cup \tau_u$ is the new set of transitions.
Repair transitions intuitively refer to good events; in our running example, e.g., the completion of a recovery.

Repair transitions cannot be added to the fault tolerance mechanism, unless we keep track of \emph{when} they are scheduled to happen, which led to significant blow-up of the state space.
But they do not represent failures; to the contrary, they represent improvements that help the fault tolerance player and should therefore not be considered in the quota allowed for the failure mechanism to win (towards an unrecoverable situation).

To illustrate this, consider our running example from Section \ref{sec:faulTolerance}, where some component recovery of a memory unit is {\em initiated} by the fault tolerance mechanism, assigning a process for that recovery. Such a component recovery cannot be modeled as immediate (atomic), because such a model would mask the competition on the free processes, leading to the trivial and unrealistic optimal strategies based on instant component recovery from every fault.
Thus, transitions which represent that a component recovery {\em terminates} should not be under the control of the fault tolerance mechanism.
On the other hand, these transitions help the fault tolerance mechanism, hence they must not come on the expense of the $k$ allowed attempts of the failure mechanism.
 
We change the game from the definition of \emph{sfrch$_k$} states as follows.
In a state $s$, the protagonist chooses a set of states contained in $\suc_c(s) \cup \suc_r(s)$ that must contain all states $\suc_r(s)=\{t \in S \mid (s,t) \in \tau_r\}$ reachable through a recovery transition.
(As recovery is an abstraction, she cannot prevent it from happening.)
The antagonist then either chooses among these states (this choice is, of course, not counted against his $k$ failure moves) or overwrites this selection by a failure transition.
Similar to strong fairness, we require that, if repair transitions are offered infinitely often, then they must be taken infinitely often.

The effect of adding the recovery transitions in this way is on the construction of the controlled limited attractor $\cla_L ( G )$.
We now get:
\\[5pt]
$\mbox{ } \qquad\cla_L ( G ) = \bigcap \{A \supseteq G \mid 
( \forall s {\in} L.\ (\suc_c(s) {\cup} \suc_r(s)) \cap A \neq 
\emptyset \mbox{ and } \suc_r(s) \subseteq A ) \mbox{ implies } s \in A\}$\\[5pt]
is the smallest set $A \supseteq G$ such that $s \in A$, if 
\begin{enumerate}
\item there is a controlled \emph{or} a recovery transition from $s$ to $A$, and
\item all recovery transitions originating from $s$ lead to $A$.
\end{enumerate}

Note that we assume here that, as part of its strategy, the recovery mechanism can decide \emph{not} to play a move, in the light of the possibility that, when she waits long enough, the repair agent will improve the situation, progressing towards $G$.

\paragraph{\bf Complexity.}
When used on graphs (as in Subsection \ref{ssec:kSafety}), then determining if a state is in the controlled limited attractor is NL-complete (cf.\ Lemma~\ref{lem:NLcla}).
This situation changes when we use the controlled limited attractor on games: the problem becomes PTIME-complete.
From a more applied point of view, however, nothing changes: the attractor construction still needs to visit each transition of $\tau_u$ and each transition in $\tau_r$ once.

\begin{lemma}
$\cla_L(G)$ can be constructed in time $O(|\cal T|)$, and checking membership of a state in $\cla_L(G)$ is PTIME-complete.
\end{lemma}

\begin{proof}
The running time is implied by the construction, while hardness can be shown by the same reduction to AND-OR graphs as in the hardness proof of Lemma~\ref{lem:ptc}, replacing fault transitions by repair transitions.
\qed
\end{proof}

The maintained bound on the running time for computing the controlled limited attractor in games implies that the bounds of $\safe_k$ and $\res_k$ are maintained as well.

\begin{corollary}
$\safe_k(G)$ resp.\ $\res_k(G)$ can be constructed in time $O(k \cdot |{\cal T}|)$ resp.\
$O(k \cdot |S| \cdot |{\cal T}|)$.

For the initial state $\iota$, we can determine the resilience level $k_{\max} =
\max\{n\in \mathbb N_0 \mid \iota \in \res_n(S \smallsetminus F)\}$ of $\iota$,
$\res_{k_{\max}}(S \smallsetminus F)$, and a memoryless $k_{\max}$-resilient control
strategy for $\res_{k_{\max}}(S\smallsetminus F)$ in $O(|S| \cdot |{\cal T}|\cdot k_{\max} \log k_{\max})$ time.
\qed
\end{corollary}

As an alternative definition, we could force the control player to make a move instead of allowing her to abstain.
In this alternative version, the evaluation of the attractor game is similar to the evaluation of a reachability game with strong fairness.
Consequently the complexities would grow by a factor of $|S|$.

\begin{theorem}
\label{theo:nested}
In the alternative setting, $\safe_k(G)$ and $\res_k(G)$ can be 
constructed in time $O(k \cdot |S| \cdot |{\cal T}|)$ and 
$O(k \cdot |S|^2 \cdot |{\cal T}|)$, respectively.

For the initial state $\iota$, we can determine the resilience level $k_{\max} =
\max\{n\in \mathbb N_0 \mid \iota \in \res_n(S \smallsetminus F)\}$ of $\iota$,
$\res_{k_{\max}}(S \smallsetminus F)$, and a memoryless $k_{\max}$-resilient control
strategy for $\res_{k_{\max}}(S\smallsetminus F)$ in $O(|S|^2 \cdot |{\cal T}|\cdot k_{\max} \log k_{\max})$ time.
\qed
\end{theorem}

The proof and the details of the construction are moved to the appendix.

\section{Tool Implementation and Experiments \label{sec.imp.exp}} 

We adopt {\em CEFSM} 
({\em communicating extended finite-state machine})~\cite{BH89} 
as a convenient language for the description of abstract models of 
our state transition systems. 
A CEFSM consists of several finite-state machines extended with shared variables 
for the modeling of shared memory and 
with synchronizations for 
the modeling of message-passing in distributed systems.  
This is justifiable 
since the fault-tolerant algorithms may themselves be subject to  
restrictions in concurrent or distributed computation.  
Indeed, we found CEFSM very expressive in modeling the 
benchmarks from the literature \cite{CL99,RSB90}.  
\smallskip

\noindent{\bf Implementation.}
In the following, we report our implementation and experiment 
with our constructions.  
Our implementation is based on symbolic on-the-fly 
model-checking techniques 
and built on the simulation/model-checking library of REDLIB in 
\texttt{http://sourceforge.net/projects/redlib/} for fast implementation.  
Our implementation and benchmarks can also be found in the same page. 

The translation from our CEFSMs to state transition systems, 
for example finite Kripke structures, is standard in the literature. 
All state spaces, conditions, preconditions, post-conditions, 
fixed points, etc. are represented as logic formulas. 
The logic formulas are then implemented with multi-value decision diagrams 
(MDD) \cite{MD98}.  
Due to the huge complexity of the transition relations, we did not explicitly construct the transition relations for precondition and 
post-condition calculation.  
Instead, the preconditions and post-conditions are constructed 
in a piecewise construction from the basic conditions for simple 
actions and transition rule triggering conditions.  
\smallskip

\noindent{\bf Benchmarks.}
We use the following five parameterized 
benchmarks to check the performance of our techniques. 
Each benchmark has parameters for the number of participating 
modules in the model.  
Such parameterized models come in handy for the evaluation of the scalability of our techniques with respect to concurrency and model sizes.  
\begin{itemize} 
\item[1.] We use the running example of avionic systems in 
  Section~\ref{sec:faulTolerance} as our first benchmark.  
  An important feature of this benchmark is that there is an 
  assumed mechanism in detecting faults of the modules.  
  Once a fault is detected, a processor can be assigned to recover the 
  module, albeit to the cost of a reduced redundancy in the executions.

\item[2.] Voting is a common technique for fault tolerance through 
  replication when there is no mechanism to detect faults of the 
  modules.  
  In its simplest form, a system can guarantee correctness, provided less than half of its modules are faulty.   

\item[3.] This is a simplified version of the previous voting benchmark, where we assume that there is a blackboard 
  for the client to check the voting result.  

\item[4.] {\em Practical Byzantine fault-tolerance} ({\em PBFT}) algorithm: 
  We use an abstract model of the famous algorithm by 
  Castro and Liskov \cite{CL99}. 
  It does not assume the availability of a fault-detection 
  mechanism but uses 
  voting techniques to guarantee the correctness of 
  computations when less than one third of the voters are faulty.  
  This algorithm has impact on the design of many protocols 
  \cite{QU,HQ,Zyzzyva,ABsTRACTs,Aardvark} and is used in 
  Bitcoin (\verb+http://bitcoin.org/+), a peer-to-peer digital currency system.  

\item[5.] {\em Fault-tolerant clock synchronization algorithm}: 
  Clock synchronization is a central issue in distributed computing. 
  In \cite{RSB90}, 
  Ramanathan, Shin, and Butler presented 
  several fault-tolerance clock synchronization algorithms 
  in the presence of Byzantine faults with high probability.  
  We use a nondeterministic abstract model of the convergence averaging 
  algorithm in their paper.  
  The algorithm is proven correct when no more than one third of the 
  local clocks can drift to eight time units from the median of all clock 
  readings.  
\end{itemize}
\smallskip

\noindent{\bf Modeling of the fault-tolerant systems.}
Appropriate modeling of the benchmarks is always important for 
the efficient verification of real-world target systems. 
Many unnecessary details can burden the verification algorithm 
and blow up the computation.  
On the other hand, sketchy models can then give too many false alarms 
and miss correct benchmarks.  
In this regard, we find that there is an interesting issue 
in the modeling of the above benchmarks.  
Replication and voting are commonly adopted techniques for 
achieving fault-tolerance and resilience.  
As can be seen, such fault-tolerant algorithms usually consist of 
several identical modules that use the same behavior templates.  
This observation implies that the identity of individual modules 
can be unimportant for some benchmarks.  
For such benchmarks, we can use counter abstraction 
\cite{ET99,Lubachevsky84} in their models.
(Details on the counter abstraction we used are provided in an appendix.)
Specifically, we found that we can use counter abstraction to 
prove the correctness of benchmarks 1, 2, and 3. 
In contrast, the PBFT and the clock synchronization algorithms use 
counters for each module to model the responses received 
from its peer modules. 
As a result, we decided not to use counter abstraction to model 
these two algorithms in this work. 

In the following, we use the running example of avionic systems 
in Section~\ref{sec:faulTolerance} to explain how we 
model a benchmark either as a plain CEFSM or with counter abstraction. 
We first present its CEFSM model template 
in Figure~\ref{fig.pm}. 
\begin{figure}[t] 
\begin{center}
\begin{picture}(0,0)%
\includegraphics{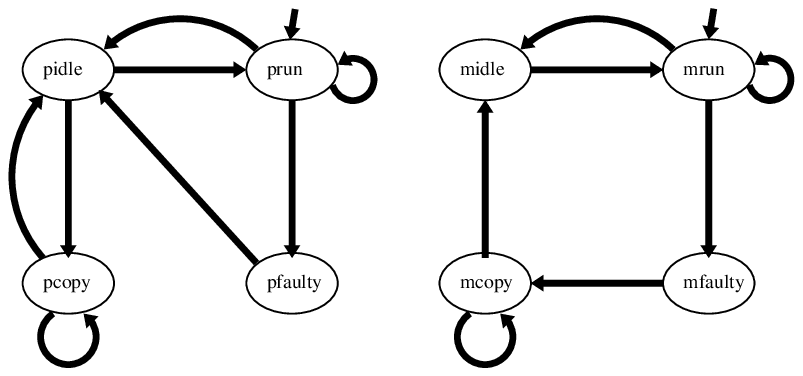}%
\end{picture}%
\setlength{\unitlength}{4144sp}%
\begingroup\makeatletter\ifx\SetFigFont\undefined%
\gdef\SetFigFont#1#2#3#4#5{%
  \reset@font\fontsize{#1}{#2pt}%
  \fontfamily{#3}\fontseries{#4}\fontshape{#5}%
  \selectfont}%
\fi\endgroup%
\begin{picture}(3923,1883)(436,-1132)
\put(1543,424){\makebox(0,0)[lb]{\smash{{\SetFigFont{6}{7.2}{\rmdefault}{\mddefault}{\updefault}{\color[rgb]{0,0,0}$C$}%
}}}}
\put(1497,-41){\makebox(0,0)[lb]{\smash{{\SetFigFont{6}{7.2}{\rmdefault}{\mddefault}{\updefault}{\color[rgb]{0,0,0}$R$}%
}}}}
\put(1055,-134){\makebox(0,0)[lb]{\smash{{\SetFigFont{6}{7.2}{\rmdefault}{\mddefault}{\updefault}{\color[rgb]{0,0,0}$C$}%
}}}}
\put(1009,-924){\makebox(0,0)[lb]{\smash{{\SetFigFont{6}{7.2}{\rmdefault}{\mddefault}{\updefault}{\color[rgb]{0,0,0}$U$}%
}}}}
\put(2078,-88){\makebox(0,0)[lb]{\smash{{\SetFigFont{6}{7.2}{\rmdefault}{\mddefault}{\updefault}{\color[rgb]{0,0,0}$U$}%
}}}}
\put(3426,400){\makebox(0,0)[lb]{\smash{{\SetFigFont{6}{7.2}{\rmdefault}{\mddefault}{\updefault}{\color[rgb]{0,0,0}$C$}%
}}}}
\put(3983,-111){\makebox(0,0)[lb]{\smash{{\SetFigFont{6}{7.2}{\rmdefault}{\mddefault}{\updefault}{\color[rgb]{0,0,0}$U$}%
}}}}
\put(2961,-134){\makebox(0,0)[lb]{\smash{{\SetFigFont{6}{7.2}{\rmdefault}{\mddefault}{\updefault}{\color[rgb]{0,0,0}$R$}%
}}}}
\put(2914,-947){\makebox(0,0)[lb]{\smash{{\SetFigFont{6}{7.2}{\rmdefault}{\mddefault}{\updefault}{\color[rgb]{0,0,0}$C$}%
}}}}
\put(3426,-552){\makebox(0,0)[lb]{\smash{{\SetFigFont{6}{7.2}{\rmdefault}{\mddefault}{\updefault}{\color[rgb]{0,0,0}$C$}%
}}}}
\put(2961,-250){\makebox(0,0)[lb]{\smash{{\SetFigFont{6}{7.2}{\rmdefault}{\mddefault}{\updefault}{\color[rgb]{0,0,0}$?$rs}%
}}}}
\put(637,-111){\makebox(0,0)[lb]{\smash{{\SetFigFont{6}{7.2}{\rmdefault}{\mddefault}{\updefault}{\color[rgb]{0,0,0}$R$}%
}}}}
\put(1218,-1087){\makebox(0,0)[lb]{\smash{{\SetFigFont{6}{7.2}{\rmdefault}{\mddefault}{\updefault}{\color[rgb]{0,0,0}(a) processors}%
}}}}
\put(3007,-1087){\makebox(0,0)[lb]{\smash{{\SetFigFont{6}{7.2}{\rmdefault}{\mddefault}{\updefault}{\color[rgb]{0,0,0}(b) memory modules}%
}}}}
\put(2333,307){\makebox(0,0)[lb]{\smash{{\SetFigFont{6}{7.2}{\rmdefault}{\mddefault}{\updefault}{\color[rgb]{0,0,0}$C$}%
}}}}
\put(4239,307){\makebox(0,0)[lb]{\smash{{\SetFigFont{6}{7.2}{\rmdefault}{\mddefault}{\updefault}{\color[rgb]{0,0,0}$C$}%
}}}}
\put(1381,284){\makebox(0,0)[lb]{\smash{{\SetFigFont{6}{7.2}{\rmdefault}{\mddefault}{\updefault}{\color[rgb]{0,0,0}\textit{crp}++;}%
}}}}
\put(2078,-181){\makebox(0,0)[lb]{\smash{{\SetFigFont{6}{7.2}{\rmdefault}{\mddefault}{\updefault}{\color[rgb]{0,0,0}\textit{crp}-\hspace*{.2mm}-;}%
}}}}
\put(2078,-273){\makebox(0,0)[lb]{\smash{{\SetFigFont{6}{7.2}{\rmdefault}{\mddefault}{\updefault}{\color[rgb]{0,0,0}\textit{cfp}++;}%
}}}}
\put(1566,-134){\makebox(0,0)[lb]{\smash{{\SetFigFont{6}{7.2}{\rmdefault}{\mddefault}{\updefault}{\color[rgb]{0,0,0}\textit{cfp}-\hspace*{.2mm}-;}%
}}}}
\put(1055,-343){\makebox(0,0)[lb]{\smash{{\SetFigFont{6}{7.2}{\rmdefault}{\mddefault}{\updefault}{\color[rgb]{0,0,0}\textit{idm}$=q$;}%
}}}}
\put(1055,-250){\makebox(0,0)[lb]{\smash{{\SetFigFont{6}{7.2}{\rmdefault}{\mddefault}{\updefault}{\color[rgb]{0,0,0}$!$fd@$q$}%
}}}}
\put(637,-204){\makebox(0,0)[lb]{\smash{{\SetFigFont{6}{7.2}{\rmdefault}{\mddefault}{\updefault}{\color[rgb]{0,0,0}$!$rs}%
}}}}
\put(451,-297){\makebox(0,0)[lb]{\smash{{\SetFigFont{6}{7.2}{\rmdefault}{\mddefault}{\updefault}{\color[rgb]{0,0,0}\textit{idm}$=0$;}%
}}}}
\put(1520,656){\makebox(0,0)[lb]{\smash{{\SetFigFont{6}{7.2}{\rmdefault}{\mddefault}{\updefault}{\color[rgb]{0,0,0}\textit{crp}-\hspace*{.2mm}-;}%
}}}}
\put(1404,656){\makebox(0,0)[lb]{\smash{{\SetFigFont{6}{7.2}{\rmdefault}{\mddefault}{\updefault}{\color[rgb]{0,0,0}$C$}%
}}}}
\put(3356,633){\makebox(0,0)[lb]{\smash{{\SetFigFont{6}{7.2}{\rmdefault}{\mddefault}{\updefault}{\color[rgb]{0,0,0}$C$}%
}}}}
\put(3449,633){\makebox(0,0)[lb]{\smash{{\SetFigFont{6}{7.2}{\rmdefault}{\mddefault}{\updefault}{\color[rgb]{0,0,0}\textit{crm}-\hspace*{.2mm}-;}%
}}}}
\put(3402,284){\makebox(0,0)[lb]{\smash{{\SetFigFont{6}{7.2}{\rmdefault}{\mddefault}{\updefault}{\color[rgb]{0,0,0}\textit{crm}++;}%
}}}}
\put(3983,-204){\makebox(0,0)[lb]{\smash{{\SetFigFont{6}{7.2}{\rmdefault}{\mddefault}{\updefault}{\color[rgb]{0,0,0}\textit{crm}-\hspace*{.2mm}-;}%
}}}}
\put(3983,-297){\makebox(0,0)[lb]{\smash{{\SetFigFont{6}{7.2}{\rmdefault}{\mddefault}{\updefault}{\color[rgb]{0,0,0}\textit{cfm}++;}%
}}}}
\put(3426,-715){\makebox(0,0)[lb]{\smash{{\SetFigFont{6}{7.2}{\rmdefault}{\mddefault}{\updefault}{\color[rgb]{0,0,0}$?$fd}%
}}}}
\put(3426,-808){\makebox(0,0)[lb]{\smash{{\SetFigFont{6}{7.2}{\rmdefault}{\mddefault}{\updefault}{\color[rgb]{0,0,0}\textit{cfm}-\hspace*{.2mm}-;}%
}}}}
\end{picture}%
\end{center}
\vspace*{-7.5mm}
\caption{State graphs of processes and memory copies}
\label{fig.pm}
\end{figure}
The CEFSM model has $n$ processors 
and $m$ memory modules. 
Figures~\ref{fig.pm}(a) and (b) are 
for the abstraction of processors and memory copies, respectively.  
The ovals represent local states of a processor or a 
memory module, while the arrows represent transitions. 
The transitions of a CEFSM  
are labeled with labels `U' (for uncontrollable), 
`C' (for controllable), or `R' (for recovery).  
For example, when a memory module moves into a faulty state, 
an idle processor may issue an {\tt fd} (fault-detected) event
and try to repair the module by copying memory contents from 
normal memory modules.  
Such fault-detection is usually achieved with standard hardware.  
Note that the benchmarks are models that reflect the 
recovery mechanism, abstracting away the details of the original systems.  
A central issue in the design of this recovery mechanism is then 
the resilience level of the controlled systems.  
More details are provided in an appendix.
\smallskip

\noindent{\bf Performance data.}
We report the performance data 
in Table~\ref{tab.perf} for the resilience algorithms 
described in Section \ref{subsec.a.ra} 
against the parameterized benchmarks in the above with various 
parameters.  
\begin{table}[t] 
\vspace*{-5mm}
\caption{Performance data for resilience calculation \hspace{21mm} s: seconds; M: megabytes} 
\label{tab.perf} 
\begin{center} 
\begin{tabular}{l|c||c|c||c|c} \hline 
benchmarks & concurrency & \multicolumn{2}{c||}{sfrch$_k$}  
			& \multicolumn{2}{c}{res$_k$} \\\cline{3-6}
	 & 	& time  & memory   & time & memory  \\
\hline \hline 
avionics & 6 processors \& 6 memory modules & 2.89s & 129M & 3.54s & 516M \\ \cline{2-6} 
	 & 7 processors \& 7 memory modules & 10.7s & 216M & 23.4s & 808M \\ \cline{2-6}
	 & 8 processors \& 8 memory modules & 43.8s & 1009M & 135s & 2430M \\ \hline 
voting	 & 1 client \& 20 replicas & 5.18s & 229M & 13.8s & 236M \\ \cline{2-6} 
	 & 1 client \& 26 replicas & 15.2s & 334M & 48.1s & 348M \\ \hline 
simple	 & 1 client \& 150 replicas & 0.71s & 159M & 31.7s & 219M \\ \cline{2-6} 
voting	 & 1 client \& 200 replicas & 1.06s & 161M & 162s & 337M \\ \cline{2-6} 
	 & 1 client \& 250 replicas & 1.36s & 163M & 307s & 499M \\ \hline 
PBFT 	 & 1 client \& 6 replicas & 0.34s & 72M & 1.05s & 193M \\ \cline{2-6} 
	 & 1 client \& 9 replicas & 13.3s & 564M & 58.5s & 1657M \\ \hline 
clock	 & 1 client \& 15 servers  &  13.3s & 547M & 25.1s & 648M  \\ \cline{2-6} 
sync	 & 1 client \& 17 severs  & 20.7s & 957M & 58.3s & 1127M  \\ \hline 
\end{tabular}
\vspace*{-.8mm}
\end{center}
\vspace*{-5mm}
\end{table}
For the avionics system,
the resilience level $k$ is set to 
one less than half the 
number of processors.  
For the voting and simple voting benchmarks, 
the value of $k$ is set to one less than half the number of 
replicas (voters). 
For the PBFT and clock synchronization algorithm, we choose 
$k$ to be one less than one third of the number of replicas. 

The performance data has been collected with a Virtual Machine (VM) 
running opensuse 11.4 x86 
on Intel i7 2600k 3.8GHz CPU 
with 4 cores and 8G memory. 
The VM only uses one core and 4G memory.  

The time and space used to calculate resilience is a little bit more 
than that to check for $\safe$.  
The reason is that $\safe_k$ is a pre-requisite for calculating $\res_k$.  
In our experiment, $\safe_k$ is usually 
very close to $\res_k$ and does not require much extra time in 
calculating $\res_k$ out of $\safe_k$.  

The experiments show that our techniques 
scale to realistic levels of redundancy.
For fault-tolerant hardware, usually the numbers of replicas are small, 
for example, less than 10 replicas. 
Thus our techniques seem very promising for the 
verification of hardware fault-tolerance.  

On the other hand, 
nowadays, software fault-tolerance through networked computers can 
create huge numbers of replicas.  
Our experiment shows that abstraction can be a useful 
techniques for the modeling and verification of software resilience.

\section{Discussion}
We introduced an approach for the development of a control of safety critical
systems that maximizes the number of \emph{dense} failures the system can tolerate.
Our techniques are inspired by the problem of controlling systems with redundancy: in order to deflect the effect of individual failures, safety critical systems are often equipped with multiple copies of various components.
If one or more components fail, such systems can still work properly as long as the correct behavior can be identified.
 
This inspired the two layered approach we developed in this paper.
In a first layer, a complex system system is simplified to the aspects relevant for its control.
This abstraction layer of our approach reduces life size examples to small abstractions.
In a second layer, we develop a control strategy, 
in which the controller identifies a $k$-resilient region; the
controller can recover without seeing an additional error to 
the $k$-resilient region.
Such a recovering strategy is memoryless.
Being memoryless on a small abstraction in particular implies that the recovery is efficient.

The system can, once recovered, tolerate and recover from $k$ further dense failures, and so forth.
Consequently, our control strategy allows for recovery from an arbitrary
number of failures, provided that the number of dense errors is restricted.
This is the best guarantee we can hope for: our technique guarantees to find the optimal parameter $k$.
This parameter is bound to be small (smaller than the number of redundant components).
Optimizing it is computationally inexpensive, but provides strong guarantees: the likelihood of having more than $k$ failures appear in short succession after a failure occurred are, for independent errors, exponential in $k$. As failures are few and far between, each level of resilience gained reduces the likelihood of errors significantly.

\providecommand\doi[1]{\newblock doi: \href{http://dx.doi.org/#1}{#1}.}

\end{document}